\documentclass[11pt,A4]{article}
\usepackage{amsmath,amsthm,nicefrac}
\usepackage{amsfonts, amstext}
\usepackage{graphicx} 
\usepackage{textcomp} 
\usepackage[font={small,it}]{caption}

\usepackage{enumitem}
\usepackage{mathrsfs}
\usepackage{amsmath,amssymb, amsthm, amstext}
\usepackage{amsfonts}
\usepackage{xcolor}
\usepackage{mathptmx}  
\usepackage{esvect}
\usepackage[linesnumbered, ruled, vlined]{algorithm2e} 

\SetKwComment{tcp}{\textit{\% }}{} 
\SetKwInput{KwIn}{Input} 
\SetKwInput{KwOut}{Output} 
\SetKw{KwTo}{to} 
\SetKw{KwRet}{return} 
\DontPrintSemicolon 
\SetAlgoLined 
\SetNlSty{}{}{:} 
\SetAlgoNlRelativeSize{-1}
\definecolor{NearBlack}{rgb}{0.1, 0.1, 0.6}
\usepackage{hyperref}
\hypersetup{colorlinks=true,%
            citebordercolor={.6 .6 .6},linkbordercolor={.6 .6 .6},%
citecolor=blue,urlcolor=black,linkcolor=NearBlack}

\newtheorem{theorem}{Theorem}[section]
\newtheorem{lemma}[theorem]{Lemma}
\newtheorem{claim}[theorem]{Claim}
\newtheorem{corollary}[theorem]{Corollary}
\newtheorem*{lemma*}{Lemma}

\usepackage{geometry}

\title{Sorting with constraints}
\author{Archit Manas}
\date{\vspace{-5ex}}

\begin{document}

\maketitle
\begin{abstract}
    In this work, we study the generalized sorting problem, where we are given a set of $n$ elements to be sorted, but only a subset of all possible pairwise element comparisons is allowed. We look at the problem from the perspective of the graph formed by the ``forbidden'' pairs, and we parametrise algorithms using the clique number and the chromatic number of this graph. We also extend these results to the class of problems where the input graph is not necessarily sortable, and one is only interested in discovering the partial order. We use our results to develop a simple algorithm that always determines the underlying partial order in $O(n^{3/2} \log n)$ probes, when the input graph is an Erdős–Rényi graph.  
\end{abstract}
\section{Introduction}
Comparison-based sorting has been extensively studied in computer science, and its complexity is well understood. However, in many practical settings, the cost of comparing a pair of objects can be non-uniform. 
 Non-uniform cost models can render even basic data processing tasks, such as finding the minimum or searching a key, surprisingly complex (see e.g. \cite{DBLP:journals/jcss/CharikarFGKRS02}). This motivates the exploration of structured comparison cost models. One such model is the monotone cost model: each element has an inherent numeric value, and the comparison costs are monotone functions of these values (see \cite{DBLP:conf/focs/GuptaK01}, \cite{DBLP:conf/soda/KannanK03}). A different line of work assumes that  comparison costs satisfy the triangle inequality, and \cite{DBLP:conf/approx/GuptaK05} gave poly-log competitive algorithms for sorting and max-finding under this framework. In the non-uniform setting, restricted models with a small number of distinct comparison costs are also of interest. 
 Recently, \cite{DBLP:conf/icalp/JiangWZ024} made significant progress for such settings. They have an $\widetilde{O}{\left(n^{1-1/(2W)}\right)}$-competitive algorithm when the input has $W$ distinct comparison costs.  A special case of this model, known as the $1$-$\infty$ cost model, occurs when certain pairs are \textit{forbidden} from being compared, while all other pairs incur a unit cost. This work focuses on this model, termed the generalized sorting problem, along with some related variants.\\\\
The input to the generalized sorting problem is an undirected graph $G(V, E)$ on $n = |V|$ vertices with $m = |E|$ edges, where an edge $(u, v) \in E$ indicates that the vertices $u$ and $v$ are comparable (these are the pairs with unit-cost for comparison). By probing the edge, we reveal a directed edge $(u, v)$ or $(v, u)$ depending on $u \prec v$ or $v \prec u$. Let $\vv{G}(V, \vv{E})$ be the underlying directed graph. Given the promise that $\vv{G}$ is an acylic graph with a Hamiltonian path, the objective is to find this path by adaptively probing the smallest number of edges in $E$. We let \textsc{probe}($u, v$) denote the oracle function (over edges $(u, v) \in E$) that returns $1$ if $u \prec v$ and $0$ otherwise. \\\\
When $G = K_n$, this becomes \textit{regular comparison-based sorting}, and it is well-known that $\Theta(n \log n)$ probes are both necessary and sufficient. When $G = K_{\frac{n}{2}, \frac{n}{2}}$, this becomes the \textit{nuts and bolts} sorting problem. Here too, it is known that $\Theta(n \log n)$ probes are both necessary and sufficient, as shown in \cite{DBLP:conf/soda/AlonBFKNO94}. In a breakthrough in \cite{DBLP:conf/focs/HuangKK11}, an algorithm with probe complexity $\widetilde{O}{(n^{1.5})}$ was presented for worst-case generalized sorting, which was later improved to $\widetilde{O}{\left(\sqrt{nm}\right)}$ in \cite{DBLP:conf/focs/KuszmaulN21}. This remains state-of-the-art for arbitrary $G$.\\\\A closely related problem is one where the guarantee of $\vv{G}$ being Hamiltonian is removed, and the objective is to determine the direction of all edges $\in E$ by adaptively probing as few of them as possible. Yet another version is the so-called Generalized Poset Sorting Problem (GPS), in which along with $G(V, E)$ we are additionally given an \textit{unknown poset} $\mathcal{P}(V, \prec)$, where each probe $(u, v)$ returns the relation between $u, v$, that is, $u \prec v, v \prec u$ or $u \nsim v$. This has been studied in \cite{DBLP:conf/icalp/JiangWZ024}.\\\\ 
As mentioned above, we term the edges $(u, v)$ not present in $G$ as \textit{forbidden}, denoting that these pairs of vertices cannot be compared. Let $H$ denote the (undirected) graph on $V$ formed by the forbidden edges. It is natural to expect that if $|H|$ is small, one can find the total order with a few probes. This is indeed the case, as shown in \cite{DBLP:conf/swat/BanerjeeR16}, where the path is found in $O((|H|+n) \log n)$ probes. In fact, the algorithm works even if the guarantee that $\vv{G}$ is Hamiltonian is removed.
\section{Preliminaries}
\label{sec:2}
In this work, we obtain algorithms for the generalized sorting problem, as well as in the scenario when the promise of a total order is removed, through two metrics of $H$ that can estimate how ``locally" small the latter is. In particular, we show the following: \begin{enumerate}
\item Let $\chi(H)$ denote the chromatic number of $H$, that is, the least positive integer $d$ so that there exists a function $f: V \mapsto [d]$ with $f(x) \ne f(y)$ for all $(x, y) \in H$. Then, $O(n \log n + n \chi(H))$ probes are sufficient to determine all edges of $\vv{G}$. We call this algorithm \hyperref[ColorSolve]{\textsc{ColorSolve}}.
\item Let $\omega(H)$ denote the clique number of $H$, that is, the size of the largest complete subgraph in $H$. Then, $O(n \omega(H) \log n)$ probes are sufficient to determine all the edges of $\vv{G}$. We call this algorithm \hyperref[CliqueSolve]{\textsc{CliqueSolve}}.
\end{enumerate}
Both the above algorithms work even without the guarantee of $\vv{G}$ containing a hamiltonian path. Due to $\omega(H) \leq \chi(H)$, it can be seen that \hyperref[CliqueSolve]{\textsc{CliqueSolve}} is more optimal (up to logarithmic factors). Although the time complexities of the algorithms are not of immediate concern to us, it can be seen that 
\begin{itemize}
\item Given the coloring of $H$ in $\chi(H)$ colors, $\hyperref[ColorSolve]{\textsc{ColorSolve}}$ runs in polynomial time.
\item \hyperref[CliqueSolve]{\textsc{CliqueSolve}} runs in polynomial time. A way to speed up the algorithm presented is briefly described at the end of \hyperref[sec:4:4]{Section 4.4}, after an analysis of the algorithm's probe complexity.
\end{itemize} \hyperref[sec:3]{Section 3} and \hyperref[sec:4]{Section 4} are dedicated to proving these results. \\\\
In \hyperref[sec:5]{Section 5}, we show that \hyperref[CliqueSolve]{\textsc{CliqueSolve}} can be used to design a algorithm that correctly sorts $G$ while using $O(n \log^2 n)$ probes with high probability, when $G$ is generated randomly with edge probability $p = \Omega(1)$. When $p$ is arbitrary, we design an algorithm that makes $\tilde{O}(n^{3/2})$ probes, depending on whether $p$ exceeds the threshold $\frac{\log n}{\sqrt{n}}$ or not. It was also shown in \cite{DBLP:conf/swat/BanerjeeR16} that $\tilde{O}(n^{3/2})$ comparisons suffice for general $G(n, p)$, but we mention our algorithm since the techniques we use are different. \\\\
As remarked earlier, these algorithms do not require the guarantee of the input graph having a hamiltonian path. When the existence of a hamiltonian path is guaranteed, i.e., in the so called \textit{stochastic generalised sorting problem}, where edges other than the hamiltonian path are added with probability $p$, an algorithm with an expected $O(n \log(np))$ of comparisons was presented in \cite{DBLP:conf/focs/KuszmaulN21}. It may be noted that our algorithm is deterministic, with the only randomness involved in the generation of the graph $G(n, p)$.
\section{Through $\chi(H)$}
\label{sec:3}
\subsection{High-level description}
Given the coloring $f:V \mapsto [k]$, we partition the vertices into $k$ color classes $S_1, \dots, S_k$. Each color class can be fully sorted using regular \textsc{MergeSort}, resulting in $k$ ``chains''. Next, to determine the relation of the chains with each other, we determine edges across two chains, which is achieved using a two-pointered algorithm \hyperref[AddEdges]{\textsc{AddEdges}}. \\\\
The idea is to use \hyperref[AddEdges]{\textsc{AddEdges}} to include a small number of edges by using few probes to ensure that the transitive closure of the created graph captures all the information from $\vv{G}$.
\subsection{The Algorithm}
Given the coloring $f: V \mapsto [k]$, we first construct the sets $S_1, S_2, \dots, S_k$ as \[S_i = \{v \in V \mid f(v) = i\}\]
It may be noted that $V = S_1 \sqcup S_2 \dots \sqcup S_k$. We call procedures \textsc{MergeSort}($S_i$) for all $i \in [k]$. We can represent the sorted orders obtained using the following chains \begin{align*}
    S_1 := a^{(1)}_1 \prec a^{(1)}_2 \prec &\dots \prec a^{(1)}_{|S_1|} \\
    S_2 := a^{(2)}_1 \prec a^{(2)}_2 \prec &\dots \prec a^{(2)}_{|S_2|} \\
    &\vdots \\
    S_k := a^{(k)}_1 \prec a^{(k)}_2 \prec &\dots \prec a^{(k)}_{|S_k|} \\
\end{align*}
where $a^{(i)}_j$ denotes the element of $S_i$ with rank $j$. \\\\ We now define a graph $A$ with edges $a^{(i)}_j \mapsto a^{(i)}_{j+1}$ for all valid pairs $(i, j)$. For each ordered pair of chains $(i, j)$ with $i \ne j$, we devise an algorithm \hyperref[AddEdges]{\textsc{AddEdges}} to add edges to $A$. It should be noted that this algorithm is \textit{not} symmetrical, and will be performed on all ordered pairs. \\\\ Here is the pseudocode for \hyperref[AddEdges]{\textsc{AddEdges}}$(i, j)$.
\begin{center}
\begin{algorithm}[H]
\label{AddEdges}
\caption{\hyperref[AddEdges]{\textsc{AddEdges}}($i, j$)}
\SetAlgoLined 
$R \gets |S_i|+1$\;
\For{$L \gets |S_j|$ \KwTo $1$}{
    $X \gets \left\{1 \leq x < R : \left(a^{(j)}_L, a^{(i)}_x \right) \in E\right\}$ \tcp{collect comparable elements in $a^{(i)}$ preceding $a^{(i)}_R$}
    \For{$x \in X$ in decreasing order}{
        \If{\textsc{probe}$\left(a^{(j)}_L, a^{(i)}_x\right) = 1$}{
            \textbf{break}\;
        }
        $R \gets x$\;
    }
    \If{$R \leq |S_i|$}{
        Add edge $\left(a^{(j)}_L, a_{R}^{(i)}\right)$ to $A$\;
    }
}
\end{algorithm}
\end{center}
\noindent
We have the chains \begin{align*}
    a^{(i)}_1 \prec &\dots \prec a^{(i)}_{|S_i|} \\
    a^{(j)}_1 \prec &\dots \prec a^{(j)}_{|S_j|}
    \end{align*}
We initialise two pointers, $R = |S_i|+1$ for the first chain, and process $L$ in decreasing order from $|S_j|$ to $1$. For each $L$, we look at the neighbours (with respect to $G(V, E)$) of $a^{(j)}_L$ in $S_i$ that are strictly before $a^{(i)}_R$. \\\\
We go one by one over these from right to left (using a variable $x$), and probe the directions of the edges $\left(a^{(j)}_L, a^{(i)}_x\right)$ until we find $a^{(i)}_x \prec a^{(j)}_L$ for the first time. In such a case, we add the edge $a^{(j)}_L \mapsto a^{(i)}_{x'}$, where $x'$ is the immediate previous value of $x$. Finally, we update $R$ to $x'$. These two steps are conveniently achieved using lines $8$ and $11$ in the pseudocode. This concludes the description of \hyperref[AddEdges]{\textsc{AddEdges}}.\\\\
Once we have called \hyperref[AddEdges]{\textsc{AddEdges}} over all pairs $(i, j)$ with $i \ne j$, we obtain a graph $A$ which consists of edges $\left(a^{(i)}_j, a^{(i)}_{j+1}\right)$ for all valid $i, j$ and those edges added in $\hyperref[AddEdges]{\textsc{AddEdges}}$. As is proved later, the edges of $\vv{G}$ can be determined from the transitive closure of $A$, that is, an edge $(u, v) \in E$ is oriented as $u \mapsto v$ in $\vv{E}$ if and only if $v$ is reachable from $u$ using the edges of $A$. \\\\
Thus, by using a series of depth first searches in $A$, we can determine $\vv{G}$ entirely without any further probes. Our final algorithm, \hyperref[ColorSolve]{\textsc{ColorSolve}} thus has the pseudocode: 
\begin{center}
\begin{algorithm}[H]
    \label{ColorSolve}
    \SetAlgoLined
    \caption{{\hyperref[ColorSolve]{\textsc{ColorSolve}}}}
    \KwIn{$G(V, E), f, k$}
    $A \gets$ empty graph on $V$\;
    \For{$i \gets 1$ \KwTo $k$}{
        $S_i \gets \{v \in V \mid f(v) = i\}$\;
        $a^{(i)} \gets \textsc{MergeSort}(S_i)$\;
        \For{$j \gets 1$ \KwTo $|S_i| - 1$}{
            Add edge $\left(a^{(i)}_j, a^{(i)}_{j+1}\right)$ to $A$ \tcp{Add the ``chain'' edges to $A$}
        }
    }
    \For{$i \gets 1$ \KwTo $k$}{
        \For{$j \gets 1$ \KwTo $k$}{
            \If{$i \neq j$}{
                \hyperref[AddEdges]{\textsc{AddEdges}}($i, j$)\;
            }
        }
    }
    $\vv{K} \gets$ empty graph on $V$\;
    \ForEach{$(u, v) \in E$}{
        \If{$\exists$ path from $u$ to $v$ in $A$}{
            Add edge $(u, v)$ to $\vv{K}$ \tcp{This may be efficiently checked using a DFS}
        }
    }
    $\vv{K}$ is the same as the underlying graph $\vv{G}$, and it has now been found.
\end{algorithm}
\end{center}
\subsection{Correctness}
In this section, we establish that \hyperref[ColorSolve]{\textsc{ColorSolve}} correctly determines $\vv{G}$. \\\\
First, we note that all edges added in $A$ are present in the underlying graph $\vv{G}$ as well. Indeed, this is true of edges of the form $\left(a^{(i)}_j, a^{(i)}_{j+1}\right)$ added through line $6$ of \hyperref[ColorSolve]{\textsc{ColorSolve}}. It is also true of edges $\left(a^{(j)}_L, a^{(i)}_R\right)$ added through line $11$ in \hyperref[AddEdges]{\textsc{AddEdges}}, since if $R \leq |S_i|$, it means that $\textsc{probe}(a^{(j)}_L, a^{(i)}_x) = 1$ was confirmed at some point. As a result, our graph $A$ is a subgraph of $\vv{G}$. \\\\
Thus, to show correctness, it suffices to show that if $(u, v)$ in $\vv{G}$ for some edge $(u, v) \in G$, then there is a path from $u$ to $v$ in $A$ (the other direction has been taken care of). 
\begin{lemma}
\label{lem:3:1}
If $(u, v)$ is an edge in $\vv{G}$, there is a path from $u$ to $v$ in $A$.
\end{lemma}
\begin{proof}
There are two cases to consider, depending on whether $f(u) = f(v)$ or not. \begin{enumerate}
    \item [(a)] If $f(u) = f(v) = i$, then it is clear that $u$ precedes $v$ in the sorted order $a^{(i)}$, and there is indeed a path from $u$ to $v$ in $A$ by simply following edges of the path \[a^{(i)}_1 \mapsto a^{(i)}_2 \mapsto \dots \mapsto a^{(i)}_{|S_i|}.\]
    \item [(b)] Suppose $j = f(u)$, $i = f(v)$, $u = a^{(j)}_s$ and $v = a^{(i)}_t$. Suppose we tried to add edge $a^{(j)}_s \mapsto a^{(i)}_T$ to $A$ through line $11$ of process \hyperref[AddEdges]{\textsc{AddEdges}}($i, j$). If no edge was added, that is, $R = |S_i|+1$ at the time, it means that $a^{(j)}_s $ was more than all the elements in $S_i$ that it could be compared with, which is impossible since $a^{(j)}_s \prec a^{(i)}_t \in S_i$. \\\\
    Thus some edge $\left(a^{(j)}_{s}, a^{(i)}_{T} \right)$ was indeed added to $T$. Now, we claim that $T \leq t$. Indeed, if this were not the case, then $t < T$. Since the edge $\left(a^{(j)}_s, a^{(i)}_{T}\right)$ was added, this means that $a^{(j)}_s \succ a^{(i)}_x$, where $x$ is the least element in set $X$ after $T$, where $X$ is the set defined through line $3$ in \hyperref[AddEdges]{\textsc{AddEdges}}($i, j$). \\\\
    This means that $x \geq t$, since $t \in X$ as well. Thus \[a^{(j)}_s \succ a^{(i)}_x \succeq a^{(i)}_t\]
    contradicting $\left(a^{(j)}_s, a^{(i)}_t\right) \in \vv{G}$. Therefore $T \leq t$, and we have the path \[u = a^{(j)}_s \mapsto a^{(i)}_T \mapsto a^{(i)}_{T+1} \dots \mapsto a^{(i)}_t = v\]
    with all edges in $A$.
\end{enumerate}
\end{proof}

\noindent
\hyperref[lem:3:1]{Lemma 3.1} and the preceding discussion prove the correctness of \hyperref[ColorSolve]{\textsc{ColorSolve}}.
\subsection{Probe complexity analysis}
In this section, we establish that \hyperref[ColorSolve]{\textsc{ColorSolve}} makes $O(n \log n + nk)$ calls to the oracle function \textsc{probe}. \\\\
Probes are either made through calls to \textsc{MergeSort} or \hyperref[AddEdges]{\textsc{AddEdges}}. 
\begin{lemma}
\label{lem:3:2}
The total number of probes made by calls to \textsc{MergeSort} is $O(n \log n)$. 
\end{lemma}
\begin{proof}
The number of probes made by \textsc{MergeSort}($S_i$) is $O(|S_i| \log |S_i|)$ which is also $O(|S_i| \log n)$. Adding over all $i$ and realising that $\sum \limits_{i} |S_i| = n$, we see that the total number of probes made due to calls to $\textsc{MergeSort}$ is $O(n \log n).$
\end{proof}
\begin{lemma}
\label{lem:3:3}
The total number of probes made by \hyperref[AddEdges]{\textsc{AddEdges}}($i, j$) is $O( |S_i| + |S_j|)$.
\end{lemma}
\begin{proof}
We show that after each call to \textsc{probe} through line 5 in \hyperref[AddEdges]{\textsc{AddEdges}}, either of $L$ or $R$ strictly decreases. Since $L+R \leq |S_i|+|S_j|+1$, this will prove the Lemma. \\\\
If \textsc{probe}$\left(a^{(j)}_{L}, a^{(i)}_x\right) = 1$, then we break and $L$ reduces by one since the iteration is finished for the current $L$. On the other hand, if \textsc{probe}$\left(a^{(j)}_{L}, a^{(i)}_x\right) = 0$, then $R$ is replaced with $x$, strictly reducing $R$. \\\\
Thus, each call to \textsc{probe} reduces $L+R$ by at least $1$, proving the lemma.
\end{proof}
\begin{corollary} \label{cor:3:4}
The total number of probes made by calls to \hyperref[AddEdges]{\textsc{AddEdges}} is $O(nk)$.
\end{corollary}
\begin{proof}
From \hyperref[lem:3:3]{Lemma 3.3}, the total number of probes made over all pairs is equal to \[\sum_{1 \leq i \leq k} \sum_{1 \leq j \leq k, j \ne i} O(|S_i|+|S_j|) \leq \sum_{1 \leq i \leq k} \left[O(k|S_i|) + \sum_{1 \leq j \leq k} O(|S_j|)\right] \leq \sum_{1 \leq i \leq k} \left[O(k|S_i|) + O(n)\right] = O(nk)\]
where we use $n = \sum \limits_{i} |S_i|$ in the last two estimates.
\end{proof}

\noindent Combining \hyperref[lem:3:2]{Lemma 3.2} and \hyperref[cor:3:4]{Corollary 3.4}, we see that the total number of probes made by \hyperref[ColorSolve]{\textsc{ColorSolve}} is $O(n \log n + nk)$, as claimed.
\section{Through $\omega(H)$}
\label{sec:4}
\subsection{High-level description}
We will set $k = \omega(H)+1$, so that $H$ has no cliques of size $k$. In other words, among any $k$ vertices, some two will be comparable. The main idea is that under this constraint, one can always find a good ``pivot'' given many vertices from the graph, via the algorithm \hyperref[Pivot]{\textsc{Pivot}}. \\\\
The next idea is to add vertices $u \in V$ one by one and look at the set of neighbours of $u$ in $G$ that have already been processsed. We find a good pivot for this set, and compare $u$ with this pivot to determine the direction of many edges containing $u$ with a single probe. Recursively repeating this algorithm until all edges are extinguished lets us determine all edges containing $u$ in a few probes. This is done using the algorithm \hyperref[DirectEdges]{\textsc{DirectEdges}}.\\\\
Combining these techniques, we are able to determine $\vv{G}$ in $O(nk \log n)$ probes.

\subsection{The Algorithm}
\hyperref[Pivot]{\textsc{Pivot}} essentially finds $u, S_{-}, S_{+}$ as described in the following lemma (the correctness is proven in \hyperref[sec:4:3]{Section 4.3}):
\begin{lemma}
\label{lem:4:1}
Let $G$ be a directed acyclic graph on $n > 10k$ vertices so that every induced subgraph of $G$ on $k$ vertices has at least one edge. Then, we can find a vertex $u$ and sets $S_{-}, S_{+}$ such that \begin{itemize}
    \item For all $x \in S_{-}$ ($x \in S_{+}$), we have $x \prec u$ ($u \prec x$).
    \item $|S_{-}|, |S_{+}| \geq \frac{n}{3k}$.
    \end{itemize}
\end{lemma}

\noindent
Before we describe \hyperref[Pivot]{\textsc{Pivot}}, we will discuss its main subprocedure \hyperref[Select]{\textsc{Select}}($G$), which finds $X$, a set of vertices of $G$ as in the following Lemma: 
\begin{lemma}
\label{lem:4:2}
Let $G$ be a directed acyclic graph on $n > 10k$ vertices so that every induced subgraph of $G$ on $k$ vertices has at least one edge. Then, we can find a set $X \subseteq V$ and sets $S_{+}(u)$ ($S_{-}(u)$) for all $u \in V$ such that the following is true: \begin{itemize}
\item $|X| > \frac{n}{2}$ and for all $u \in X$, we have $|S_{+}(u)| \geq \frac{n}{3k}$ ($|S_{-}(u)| \geq \frac{n}{3k}$).
\item For all $u \in V$ and $x \in S_{+}(u)$ ($x \in S_{-}(u)$), we have $u \prec x$ ($x \prec u$).
\end{itemize}
\end{lemma}

\begin{center}
\begin{algorithm}[H]
\label{Select}
\caption{\hyperref[Select]{\textsc{Select}($G$)}}
\SetAlgoLined
\KwIn{$G(V, E), k$ \tcp{$|V| > 10k$, and $G$ has no empty induced subgraph on $k$ vertices}}
$n \gets |V|$\;
$\mathcal{T} \gets \{\}$\;
\For{$u \in V$}{
    $S_{+}(u) \gets \{\}$ \tcp{Initialise sets $S_{+}(u)$}
    Add $\{u\}$ to $\mathcal{T}$ \tcp{Add $1$-vertex tree $\{u\}$ to $\mathcal{T}$}
}
$R \gets \lceil \frac{n}{3k} \rceil$\;
$Z \gets V$\;
\For{$i \gets 1$ \KwTo $R$}{
    \While{$ \exists ~ T_1, T_2 \in \mathcal{T}$ such that $(\text{root}(T_1), \text{root}(T_2)) \in G$}{
        $\mathcal{T} \gets \mathcal{T} \setminus \{T_1, T_2\}$ \tcp{remove the two ``mergeable'' trees temporarily}
        \If{$\text{root}(T_1) \prec \text{root}(T_2)$}{
            Parent of $\text{root}(T_1) \gets \text{root}(T_2)$ \tcp{link the two trees as in \hyperref[select:a]{(a)}}
            Add $T_2$ to $\mathcal{T}$ \tcp{add the merged tree back to $\mathcal{T}$}
        }
        \Else{
            Parent of $\text{root}(T_2) \gets \text{root}(T_1)$\;
            Add $T_1$ to $\mathcal{T}$\;
        }
    }
    $S_r \gets \{\}$\;
    \For{tree $T \in \mathcal{T}$}{
        $r \gets \text{root}(T)$\;
        Add $r$ to $S_r$\;
        \For{$u \in T \setminus \{r\}$}{
            Add $r$ to $S_{+}(u)$\;
        }
    }
    $Z \gets Z \setminus S_r$\;
    Delete vertices in $S_r$ from $\mathcal{T}$ \tcp{Roots of all trees in $\mathcal{T}$ are deleted}
}
We now have the set of vertices $Z$ and sets $S_{+}(\cdot)$\;
\end{algorithm}
\end{center}

\noindent
We note that by calling $\hyperref[Select]{\textsc{Select}}($G$)$ on the graph obtained by reversing the edges of $G$, we can replace $S_+$ in the above with $S_{-}$ everywhere.\\\\
\hyperref[Select]{\textsc{Select}} proceeds by maintaining a set $Z$ of \textit{active} vertices, and a partition of $Z$ into rooted trees $\mathcal{T} = T_1 \sqcup T_2 \dots \sqcup T_m$, with each $T_i$ having edges directed towards its root. All the edges in these trees are also maintained to be edges in $G$. \\\\
Initially, $Z = V$, and $\mathcal{T}$ consists of $n$ 1-vertex trees, one for each vertex $u \in V$. Set $R = \lceil \frac{n}{3k} \rceil$. The algorithm proceeds for $R$ rounds. In each round, the following events take place: 
\begin{enumerate}
\item [(a)] \label{select:a} While there are two trees $T_1, T_2 \in \mathcal{T}$ such that there is an edge between $\text{root}(T_1)$ and $\text{root}(T_2)$, we compare the two, and depending on whether $\text{root}(T_1) \prec \text{root}(T_2)$ or $\text{root}(T_2) \prec \text{root}(T_1)$ we link $\text{root}(T_1)$ as a child of $\text{root}(T_2)$ or vice versa, merging the two trees into a bigger tree.
\item [(b)] If no two trees may be merged as in (a), we look at the set of roots of the trees in $\mathcal T$. For all $T_i \in \mathcal T$, if $r_i$ is the root of tree $T_i$, we add $r_i$ to $S_{+}(u)$ for all $u \in T_i \setminus \{r_i\}$. Then all vertices $r_1, r_2, \dots, r_m$ are removed from $Z$ and their trees, rendering them inactive. Note that this step leads to some of the trees in $T_i$ splitting into multiple smaller trees, which might be merged again as detailed in (a).
\end{enumerate}
After the rounds are over, the set of alive vertices $Z$ is returned as $X$, alongisde the desired sets $S_{+}(u)$.

\noindent
The algorithm for \hyperref[Pivot]{\textsc{Pivot}} can now be developed using the procedure \hyperref[Select]{\textsc{Select}}. We call \hyperref[Select]{\textsc{Select}} on our original graph $G$ and $G$ reversed, thereby obtaining sets $X_1, X_2$ and collections $S_{+}(\cdot), S_{-}(\cdot)$ satisfying the following properties (as will be proven in \hyperref[sec:4:3]{Section 4.3}): 
\begin{itemize}
    \item $|X_1|, |X_2| > \frac{n}{2}$.
    \item $u \prec x$ for all $x \in S_{+}(u)$ and $x \prec u$ for all $x \in S_{-}(u)$.
    \item $|S_{+}(u)| \geq \frac{n}{3k}$ for all $u \in X_1$.
    \item $|S_{-}(u)| \geq \frac{n}{3k}$ for all $u \in X_2$.
\end{itemize}
Since each of $X_1$ and $X_2$ have more than half the vertices of $G$, they must have at least one vertex in common, say $u$. Now \hyperref[Pivot]{\textsc{Pivot}} may simply return $u, S_{-}(u)$ and $S_{+}(u)$.
\begin{center}
\begin{algorithm}[H]
\label{Pivot}
\caption{\hyperref[Pivot]{\textsc{Pivot}}($G$)}
\KwIn{$G(V, E), k$ \tcp{$|V| > 10k$, and $G$ has no empty induced subgraph on $k$ vertices}}
$Z_1, S_{+}(\cdot) \gets \hyperref[Select]{\textsc{Select}}(G)$\;
$G' \gets$ reversed $G$ \tcp{$G'$ is the graph $G$ with all edges reversed}
$Z_2, S_{-}(\cdot) \gets \hyperref[Select]{\textsc{Select}}(G')$\;
Pick $u$ from $Z_1 \cap Z_2$\;
$u, S_{+}(u), S_{-}(u)$ satisfy the desired properties\;
\end{algorithm}
\end{center}
\noindent
Next, we describe $\hyperref[DirectEdges]{\textsc{DirectEdges}}$. \\\\
$\hyperref[DirectEdges]{\textsc{DirectEdges}}$ takes as input a vertex $u$, a set $S$ of vertices that $u$ is connected to (in $G(V, E)$) such that the directions of all edges of $G$ induced by $S$ are known. It then discovers the directions of edges $(u, x)$ for all $x \in S$ and adds these edges to $\vv{G}$.\\\\
This is the procedure for \hyperref[DirectEdges]{\textsc{DirectEdges}}: 
\begin{enumerate}
    \item [(a)] If $|S| \leq 10k$, we make $|S|$ probes via $\textsc{probe}(u, x)$ for all $x \in S$, thereby determining all the directions.
    \item [(b)] When $|S| > 10k$, we find a vertex $p$ and sets $S_{+}, S_{-}$ as in \hyperref[lem:4:1]{Lemma 4.1} by calling \hyperref[Pivot]{\textsc{Pivot}}($G[S]$), where $G[S]$ denotes the subgraph of $G$ induced by $S$. 
    \item [(c)] We probe the edge $(u, p)$. If $u \prec p$, then edges $(u, x)$ are added to $\vv{G}$ for all $x \in S_{+}$, and if $u \succ p$, then edges $(x, u)$ are added to $\vv{G}$ for all $x \in S_{-}$.
    \item [(d)] Depending on whether $u \prec p$ or $u \succ p$, we recurse to either $\hyperref[DirectEdges]{\textsc{DirectEdges}}(u, S \setminus (S_{+} \cup \{p\}))$ or recurse to $\hyperref[DirectEdges]{\textsc{DirectEdges}}(u, S \setminus (S_{-} \cup \{p\}))$. 
\end{enumerate}
Here is the pseudocode for $\hyperref[DirectEdges]{\textsc{DirectEdges}}$.
\begin{center}
\begin{algorithm}[H]
\label{DirectEdges}
\caption{\hyperref[DirectEdges]{\textsc{DirectEdges}}($u$, $S$)}
\SetAlgoLined
\While{$|S| > 10k$}{
    $p, S_{+}, S_{-} \gets \hyperref[Pivot]{\textsc{Pivot}}(G[S])$ \tcp{find a good pivot as in Lemma 4.1}
    $S \gets S \setminus \{p\}$\;
    \If{\textsc{probe}$(u, p) = 1$}{
        add $(u, p)$ to $\vv{G}$\;
        \For{$x \in S_{+}$}{
            add $(u, x)$ to $\vv{G}$\;
        }
        $S \gets S \setminus S_{+}$\;
    }
    \Else{
        add $(p, u)$ to $\vv{G}$\;
        \For{$x \in S_{-}$}{
            add $(x, u)$ to $\vv{G}$\;
        }
        $S \gets S \setminus S_{-}$\;
    }
}
\For{$x \in S$}{
    \If{\textsc{probe}(u, x) = 1}{
        add $(u, x)$ to $\vv{G}$\;
    }
    \Else{
        add $(x, u)$ to $\vv{G}$\;
    }
}
Directions of all edges $(u, x)$ for $x \in S$ have been determined and added to $\vv{G}$.\;
\end{algorithm}
\end{center}
\noindent
With the above, we are ready to describe $\hyperref[CliqueSolve]{\textsc{CliqueSolve}}$:
\begin{enumerate}
    \item [(a)] We process vertices one by one, and maintain that we know all the edges in $\vv{G}[P]$, where $P$ is the set of processed vertices.
    \item [(b)] When we are processing $u$, we look at $u$'s neighbours in $G(V, E)$ that have already been processed. Let the set of these vertices be $S$. We call $\hyperref[DirectEdges]{\textsc{DirectEdges}}(u, S)$ and determine the directions of all edges $(u, x)$ with $x \in S$.
    \item [(c)] By the time we have processed all the vertices, we have all the edges in $\vv{G}$.
\end{enumerate}
\begin{algorithm}[H]
\caption{\hyperref[CliqueSolve]{\textsc{CliqueSolve}}}
\label{CliqueSolve}
\SetAlgoLined
\KwIn{$G(V, E), k$}
$\vv{K} \gets $ Graph$\{V, \{\}\}$ \tcp{Initialise an empty directed graph on $V$}
$P \gets \{\}$\;
\For{$u \in V$}{
    $S \gets \{v \in P \mid (u, v) \in G\}$\;
    \hyperref[DirectEdges]{\textsc{DirectEdges}}$(u, S)$\;
    $P \gets P \cup \{u\}$ \tcp{add $u$ to the set of processed vertices}
}
$\vv{K}$ is now the same as $\vv{G}$.\;
\end{algorithm}

\subsection{Correctness}
\label{sec:4:3}
In this section we establish that \hyperref[CliqueSolve]{\textsc{CliqueSolve}} correctly determines $\vv{G}$. \\\\
We first prove \hyperref[lem:4:2]{Lemma 4.2}, and verify that \hyperref[Select]{\textsc{Select}} indeed returns a set $X$ and a collection $S_{+}(\cdot)$ as described in the Lemma. \\\\
Note that after each round, the size of $S_{+}(u)$ for any $u \in Z$ increases by exactly one, since the root of $u$ during the start of the round in question is added to $S_{+}(u)$. In particular, by the end, for all $u \in Z$ we have $|S_{+}(u)| \geq R \geq \frac{n}{3k}$. \\\\
Thus, it suffices to show that $|Z| > \frac{n}{2}$ by the end of the $R$ rounds. 
\begin{claim}
When the while loop of \hyperref[Select]{\textsc{Select}} terminates on line 20, there are less than $k$ trees in $\mathcal T$. In particular, less than $k$ vertices are deleted from $Z$ through line $29$.
\end{claim}
\begin{proof}
We note that the while loop must necessarily terminate, since the number of trees in $\mathcal T$ reduces by one every iteration. Moreover, when the process terminates, the set of roots of the trees must form a clique in $H$, since every two of the roots are incomparable. Since cliques in $H$ have size less than $k$, it follows that there are less than $k$ trees in $\mathcal{T}$. 
\end{proof}

\noindent
With the above claim, $|Z|$ reduces by at most $k$ every round, and thus after $R$ rounds, the final set $Z$ satisfies: \[|Z| \geq |V| - k R = n - k \left\lceil \frac{n}{3k} \right \rceil > n - k \left(\frac{n}{3k} + 1\right) = \frac{2n}{3} - k > \frac{n}{2},\]
where the last estimate follows from $n > 10k$. \\\\
Having proven the correctness of \hyperref[Select]{\textsc{Select}}, the correctness of \hyperref[Pivot]{\textsc{Pivot}} is immediate. As we have already seen, the two sets $Z_1, Z_2$ obtained each contain more than half the vertices of $G$, and thus have a common element, say $u$. Then $u$ satisfies $|S_{-}(u)|, |S_{+}(u)| \geq \frac{n}{3k}$ due to the two calls to \hyperref[Select]{\textsc{Select}}. \\\\
Next, we show that \hyperref[DirectEdges]{\textsc{DirectEdges}} correctly determines the directions of edges $(u, x)$ for all $x \in S$. \\\\
Due to symmetry, we assume that the condition in line $4$ of \hyperref[DirectEdges]{\textsc{DirectEdges}} is true. Then, the edge $(u, p)$ is correctly added to $\vv{G}$ since we know $u \prec p$ to be true. Moreover, edges $(u, x)$ for $x \in S_{+}$ are correctly determined as well since for any $x \in S_{+}$ we have \[u \prec p \prec x\] 
Thus all edges added by $\hyperref[DirectEdges]{\textsc{DirectEdges}}$ during the while loop in line $1$ are correctly oriented. This is also easily seen to be the case for edges added via the for loop on line $19$. Therefore \hyperref[DirectEdges]{\textsc{DirectEdges}} correctly determines directions of all the edges $(u, x)$ with $x \in S$. \\\\
Now note that \hyperref[CliqueSolve]{\textsc{CliqueSolve}} calls \hyperref[DirectEdges]{\textsc{DirectEdges}} for each edge $(u, v)$ of $G$ exactly once, depending on the order in which $u$ and $v$ are processed. Thus, \hyperref[CliqueSolve]{\textsc{CliqueSolve}} correctly determines $\vv{G}$.
\subsection{Probe complexity analysis}
\label{sec:4:4}
In this section, we estbalish that \hyperref[CliqueSolve]{\textsc{CliqueSolve}} makes $O(nk \log n)$ calls to the oracle function \textsc{probe}.\\\\
We note that \hyperref[CliqueSolve]{\textsc{CliqueSolve}} makes calls to \textsc{probe} only via the procedure \hyperref[DirectEdges]{\textsc{DirectEdges}}. And thus we show
\begin{lemma}
\label{lem:4:4}
    The procedure \hyperref[DirectEdges]{\textsc{DirectEdges}}$(u, S)$ makes $O\left(k \log |S|\right)$ probes.
\end{lemma}
\begin{proof}
    To prove this we will look at the decrease in the size of $S$ after every iteration of the while loop on line 1 in \hyperref[DirectEdges]{\textsc{DirectEdges}}. We see that we remove $p$ from $S$, as well as either $S_{+}$ or $S_{-}$.  Due to \hyperref[lem:4:1]{Lemma 4.1}, both sets have size at least $\dfrac{|S|}{3k}$. Thus, the size of $S$ decreases by a factor of at least $\eta = \left(1 - \dfrac{1}{3k}\right)$ after every probe made in the while loop. \\\\
    Further, when the while loop is exited, no more than $10k$ additional probes are made, since $|S|$ has at most $10k$ vertices by the time we exit. Thus, the total number of probes is no more than \[10k + \log_{(1/\eta)}(|S|) = 10k + \frac{\log |S|}{\log(1/\eta)}\]
    Now \[\log(1/\eta) = \log\left(1 + \frac{1}{3k-1}\right) = \frac{1}{3k-1} \log \left(\left(1+ \frac{1}{3k-1}\right)^{3k-1}\right) \geq \frac{\log 2}{3k-1} = \Omega\left(\frac{1}{k}\right)\]
    and thus the total number of probes is \[\leq 10k + O(k \log |S|) = O(k \log |S|),\]
    proving the lemma.
\end{proof}

\noindent
Observe that the procedure $\hyperref[DirectEdges]{\textsc{DirectEdges}}(u, S)$ is called exactly $n$ times, and since each time $|S| \leq n$, the total number of probes is $O(nk \log n)$, as claimed. \\\\
Here we remark that it is possible to improve the naive complexity of the above algorithm to a total complexity of $O(m \omega(H) \log n + n^2)$ by developing an $O(|S| \log |S|)$ algorithm to replace \hyperref[Select]{\textsc{Select}} by altering the Binomial Heap to allow up to $\omega(H)$ trees at the same level. We do this by ensuring that that each $T \in \mathcal{T}$ is a Binomial tree.  Here is a brief sketch of the algorithm:
\begin{itemize}
\item Instead of combining until the roots of trees in $\mathcal{T}$ form a clique in $H$, we combine only until the roots of trees in each of the levels form a clique in $H$. Using the size of $\mathcal{T}$ as a potential function, we can show that the amortised complexity of this step is $O(k \log n)$.
\item At a given stage, there are at most $k$ trees in a level. This means there are at most $k \log n$ roots, where $n = |S|$. We can repeatedly compare the roots while there are more than $k$ of them and find up to $k$ vertices to remove. This requires $O(k \log n)$ operations. Here we do not actually change $\mathcal{T}$ to reflect these comparisons, we only do it to shortlist the (up to) $k$ roots for deletion. 
\item Once we find the (up to) $k$ roots, we delete them all.
\item We perform this $O(n/k)$ times, and thus the total complexity for \hyperref[Select]{\textsc{Select}} is $O(n \log n)$.
\end{itemize}
Since \hyperref[Select]{\textsc{Select}} can be altered to run in $O(|S| \log |S|)$ time, it follows that \hyperref[DirectEdges]{\textsc{DirectEdges}} can now run in $O(k |S| \log |S|)$ time. Summing over the vertices, we see that the total complexity of \hyperref[CliqueSolve]{\textsc{CliqueSolve}} can be made $O(mk \log n + n^2)$.
\section{Concluding Remarks}
\label{sec:5}
In this paper, we consider the generalized sorting problem and provide two algorithms that perform well when the graph $H$ formed by the forbidden edges is ``locally small'', by considering the metrics $\chi(H)$ and $\omega(H)$. In particular, we find algorithms that make $O(n\log n + n\chi(H))$ and $O(n \omega(H)\log n)$ probes and determine all edges of the input graph $G(V, E)$. \\\\
We now briefly discuss the scenario when $G$ is randomly generated with probability $p$. It is well known that the independence number of $G$ is $O(p^{-1} \log n)$ with high probability in this case, and thus $\omega(H) = O(p^{-1} \log n)$ which means that, with high probability, using \hyperref[CliqueSolve]{\textsc{CliqueSolve}} will make $O(p^{-1} n \log^2n)$ comparisons to recover $\vec{G}$. In particular, when $p = \Omega(1)$, \hyperref[CliqueSolve]{\textsc{CliqueSolve}} sorts $G$ in $O(n\log^2n)$ probes. \\\\
Actually, since $|E(G)| = O(pn^2)$ with high probability as well, we see that depending on whether $p < \frac{\log n}{\sqrt{n}}$ or not, we can use a bruteforce or \hyperref[ColorSolve]{\textsc{ColorSolve}} to obtain a final complexity of $O(n^{3/2} \log n)$.
\section{Acknowledgements}
I express my deep gratitude to Prof. Amit Kumar for introducing me to the problem and for his constant guidance and insightful comments throughout the course of this work.
{\small
\bibliographystyle{alpha}
\bibliography{citations}
}
\end{document}